\documentclass[12pt]{article}%
\usepackage{eurosym}
\usepackage{bbm}
\usepackage{amsmath}
\usepackage{amsfonts}
\usepackage{amssymb}
\usepackage{graphicx}
\usepackage{pstricks}
\providecommand{\U}[1]{\protect\rule{.1in}{.1in}}
\providecommand{\U}[1]{\protect\rule{.1in}{.1in}}
\newcommand{\N}{{\mathbb N}}

\newcommand{\R}{{\mathbb R}}

\newcommand{\Pc}{{\mathcal P}}

\def\rpp0{\rho_{\Pc_0}}
\def\rpp1{\rho_{\Pc_1}}

\newcommand{\ba}{\begin{eqnarray}}
\newcommand{\ea}{\end{eqnarray}}
\newcommand{\bas}{\begin{eqnarray*}}
\newcommand{\eas}{\end{eqnarray*}}
\newcommand{\be}{\begin{equation}}
\newcommand{\ee}{\end{equation}}
\newcommand{\bi}{\begin{itemize}}
\newcommand{\ei}{\end{itemize}}

\newtheorem{theorem}{Theorem}
\newtheorem{proposition}[theorem]{Proposition}

\newtheorem{definition}[theorem]{Definition}

\newenvironment{proof}[1][Proof]{\noindent\textbf{#1.} }{\ \rule{0.5em}{0.5em}}
\newtheorem{preremark}[theorem]{Remark}
\newenvironment{remark}{\begin{preremark}\rm}{\hfill$\Diamond$\end{preremark}}
\newtheorem{prenotation}[theorem]{Notation}

\numberwithin{equation}{section}
\numberwithin{theorem}{section}
\begin{document}

\title{{Holomorphic fractional Fourier transforms}}
\author{William D. Kirwin\thanks{Center for Mathematical Analysis, Geometry and Dynamical Systems, Instituto Superior T\'ecnico, University of Lisbon, will.kirwin@gmail.com}, Jos\'e  Mour\~ao, Jo\~ao P. Nunes\thanks{Department of Mathematics and Center for Mathematical Analysis, Geometry and Dynamical Systems, Instituto Superior T\'ecnico, University of Lisbon, jmourao \& jpnunes@math.tecnico.ulisboa.pt} \, and Thomas Thiemann\thanks{Lehrstuhl f\"ur Theoretische Physik III, FAU Erlangen-N\"urnberg, thomas.thiemann@gravity.fau.de }}
\maketitle

\date

\begin{abstract}
The Fractional Fourier Transform (FrFT) has widespread applications in areas like signal analysis, Fourier optics, diffraction theory, etc.
The Holomorphic Fractional Fourier Transform (HFrFT) proposed in the present paper
may be used in the same wide range of applications with improved properties.
The HFrFT of  signals spans a one-parameter family
of (essentially) holomorphic functions, where the parameter takes values in the bounded interval $t\in (0,\pi/2)$. At the boundary values of the parameter, one obtains  the original signal at $t=0$ and its Fourier transform at the other end of the interval $t=\pi/2$.  If the initial signal is $L^2 $, then, for an appropriate choice of inner product that will be detailed below, the transform is unitary for all values of the parameter in the interval. This transform provides a heat kernel smoothening of the 
signals while preserving unitarity for $L^2$-signals and continuously interpolating between the original signal and its Fourier transform. 
\end{abstract} 

\tableofcontents

\section{Introduction and Results}
\label{s1}

In the present paper, we propose an holomorphic version of the Fractional Fourier Transform  which will be a one-parameter family of functional transforms that interpolate between a function in $L^2(\R,dx)$ and its Fourier transform. At intermediate values of the parameter the transform will yield, up to a multiplicative common factor, holomorphic functions on the complex plane.
 
Recall that the Fractional Fourier transform (FrFT) \cite{Co} is a group of unitary transforms, 
$${\mathcal F}_t \, : L^2(\R) \longrightarrow  L^2(\R),\,\, t\in \R,$$ 
periodic in $t$ with period $\pi$,
interpolating continuously from the identity at values $t=0$, to the Fourier transform at 
$t=\frac{\pi}{2}$ and satisfying the group property,
$$
{\mathcal F}_t \circ {\mathcal F}_s = {\mathcal F}_{t+s}   \, .
$$
The Fractional Fourier Transform (FrFT) has
widespread applications in many  areas of physics and engineering like signal analysis, Fourier optics, diffraction theory, etc. (see e.g. \cite{Al, Co, DeB, OKM, Z}).

Our contribution in this work is two-fold:

\begin{itemize}
\item[i)] Propose a family of holomorphic Fractional Fourier Transforms (HFrFT), $A_t, t\in [0,\frac{\pi}{2}]$, which
may be used in the same wide range of applications as the FrFT but with improved properties,
as the HFrFT of  signals span a one-parameter family
of holomorphic functions, with the original signal and its Fourier transform
being the boundary values of the family, at $t=0$ and $t=\frac{\pi}{2}$ respectively. If the initial signal is $L^2 $ then the transform
is unitary for all values of the parameter in the interval.
 This transform provides a heat kernel smoothening of the 
signals and of their Fourier transforms while preserving 
unitarity for $L^2$-signals and continuously interpolating between the
original signal and its Fourier transform. 

\item[ii)] We show that the family of HFrFT transforms includes the original Segal--Bargmann transform, which 
is attained at the value $t=\pi/4.$
The Segal-Bargmann transform has been applied to signal analysis
before (see, for example, \cite{SV}), where the use of the machinery of holomorphic function theory allows for nice inversion formulas and for the definition of numerically useful truncations. 
Recall that the original Segal-Bargmann transform is, in fact, a particular member of a natural continuous 
one-parameter family of holomorphic transforms, $SB_{s}$, labelled by $s>0.$ While this fact has apparently not been used 
in signal analysis, its relevance in that context is indicated by the fact that as $s\to +\infty$ the transforms $SB_{s}$ approach the Fourier transform \cite{KW}. 

In fact, in Theorem \ref{equalitytransf}, we show that for $s= \tan (t)$, and up to a signal-independent factor, $A_t$ coincides with  $SB_s$. 
In the geometric context of \cite{KMNT} where the transforms $A_t$ arise and that we very briefly describe below, this corresponds to a highly non-trivial operator identity.
Therefore, it becomes very natural, in the context of the Segal-Bargmann transform, to consider the parameter $t = \arctan{(s)}$ taking values in $t\in [0,\frac{\pi}{2}]$, such that at the finite value $t=\frac{\pi}{2}$ the Fourier transform is attained. The family of holomorphic 
transforms that we describe below corresponds, in fact, to this reparametrization of the family of Segal-Bargmann transforms, making the fact that they are HFrFT much more suggestive.
\end{itemize}

Therefore, while the usual FrFT takes $L^2$ functions to $L^2$ functions, the HFrFT that we are proposing maps $L^2$ signals to functions which, for 
$t\in(0,\frac{\pi}{2})$ and up to a common factor, are holomorphic in the complex plane. At $t=\frac{\pi}{2}$ we recover the Fourier transform of the $L^2$ signal.

The following observations will not be used in the main text but we include them to explain the context in which the HFrHF arises.
The original motivation for introducing the HFrFT \cite{KMNT} comes from the field of geometric quantization. In fact, by examining explicit formulae, for example in \cite{Co,DeB}, one can check that the FrFT can be naturally interpreted as  relating quantizations of the symplectic plane
with wave functions (essentially) depending on rotated (polarized, in the terminology of geometric quantization) arguments
\begin{equation}
\label{e-frft-gq}
\psi_0(x) \mapsto  \tilde \psi_t(x,p) = \psi_t(\cos(t) x + \sin(t) p)    \, .
\end{equation}
This thansform includes a functional change $\psi_0 \mapsto \psi_t$
and a rotation of the initial argument $x$ by
the flow $\varphi_t$ generated by the harmonic oscillator Hamiltonian,
$H= \frac12 (p^2+x^2)$,
\begin{equation}
\label{e-hosc-gq}
\varphi_t^*(x) = \cos(t) x + \sin(t) p   \, .      
\end{equation}
 For $t= \pi/2$ one gets the usual Fourier transform,
while, for instance, values $t_k = \pi / 2k$ are interpreted as $k$th roots of the Fourier
transform. The level sets of $\psi_t$ are lines which rotate counter-clockwise with $t$, starting with vertical lines parallel to 
$p-$axis at $t=0$, which corresponds to the $x-$dependent wave functions of usual Schr\"odinger quantization, and becoming horizontal lines parallel to the $x-$axis at $t=\frac{\pi}{2}$ as appropriate for the $p-$dependent wave functions of quantization in momentum space. The 
different axis in the $(x,p)$ plane which appear in the FrFT are therefore generated by the 
Hamiltonian flow of the harmonic oscillator. The explicit form of the FrFT then follows 
by applying methods of geometric quantization (for example as in \cite{FMMN,KW,KMN13, KMNT}).

By considering Hamiltonian flows analytically continued to complex time \cite{Th96,Th07,HK,MN15}, 
or flows associated with complex Hamiltonians, one can deform continuously 
wave functions depending on $x$, i.e. in Schr\"odinger quantization,
to wave functions depending holomorphically on a complex coordinate on the plane. These holomorphic wave functions 
give Hilbert spaces of so-called coherent states and, for quadratic Hamiltonians, they satisfy appropriate $L^2$ conditions on the plane, ensuring that 
the deformation is unitary. For the HFrFT that we propose in this paper, we consider the analytic continuation to complex time of the 
Hamiltonian flow of the hyperbolic Hamiltonian $H= \frac12 (p^2-x^2)$ \cite{KMNT}.
The real time evolution of the function $x$ under the flow $\widetilde \varphi_t$ 
generated by this Hamiltonian is 
\begin{equation}
\nonumber
\label{ee-hyp}
\widetilde \varphi_t^*(x) = \cosh(t) x + \sinh(t) p   \, .      
\end{equation} 
We see that by rotating $t$ to the imaginary axis we
get
\begin{equation}
\label{ee-hyp-it}
\widetilde \varphi_{-it}^*(x) = \cos(t) x + i \sin(t) p   \, ,      
\end{equation} 
a periodic rotation very similar to the one of (\ref{e-hosc-gq})
but with the crutial difference that it takes place on the $(x, ip)$--plane
instead of the $(x,p)$--plane. For $t \in (0, \pi /2)$ the 
variable $ \cos(t) x + i \sin(t) p$ is associated with a 
K\"ahler polarization with holomorphic wave functions,
while for $t = 0 $ and $t=\pi /2$ the polarizations
are real and coincide with the Schr\"odinger
(functions of $x$) and the momentum representation (functions of $p$).

\section{Holomorphic fractional Fourier transforms}

Consider the usual inner product\footnote{The factor of $\sqrt{\pi}$ is inserted for convenience.} in $L^2(\R,dx)$ normalized as
\begin{equation}\label{inner}
\langle f_1,f_2\rangle = \sqrt{\pi} \int_\R \bar f_1(x) f_2(x) dx.
\end{equation}

As in \cite{GS12, GKRS, KMNT}, let us consider an overcomplete 
system of normalized Gaussian coherent states in $L^2(\R, dx)$
\be
\label{42.gau}
\psi_Y(x) = {\pi}^{-\frac12} \, e^{-i P(x-Q) - \frac 12 (x-Q)^2} \,   ,  \qquad  Y=(P, Q) \in \R^2 .
\ee

The HFrFT defined below will take values in a space of functions on the $(x,p)-$plane which, up to an overall common factor, are holomorphic with respect 
to appropriate complex structures, as follows.
For $t\in (0,\frac{\pi}{2})$, let us consider the holomorphic cordinate on the plane
$$
w_t = \cos(t) x + i \sin(t) p.
$$
Let us consider the measure
$$
d\mu_t=  \frac{\sqrt{\sin(2t)}}{2}\,  dx\,dp,
$$
the Hilbert space $L^2(\R^2,d\mu_t)$
and let
$$
F_t(x,p) = e^{\frac{(w_t-\bar w_t)^2}{4 \sin{(2t)}}} = e^{-\frac{\tan (t)}{2} p^2}.
$$
Denoting by $Hol_t$ the space of complex-valued $w_t$-holomorphic functions on the plane, 
consider the vector space of functions on the plane
$$
V_t = \{ F_t\cdot f\,\vert\, f\in Hol_t\}.
$$
We then have the following Hilbert space
$$
{\mathcal H}_t = V_t \cap L^2(\R, d\mu_t).
$$
Therefore, elements in ${\mathcal H}_t$ are $w_t$-holomorphic functions on the plane, up to a multiplicative Gaussian factor $F_t$.

Let us now define the HFrFT\footnote{In \cite{KMNT}, we used the notation $U_\tau$ with $\tau = it.$}.

\begin{definition}
For $t\in (0,\frac{\pi}{2})$, the HFrFT 
$$
A_t: L^2(\R,dx)\to {\mathcal H}_t,
$$ 
is defined, in the overcomplete basis of coherent states (\ref{42.gau}), by 
\begin{eqnarray} \label{hfrft}\left(A_{t} \, \psi_Y \right) (x, p) &=&   \alpha_t \, e^{-i P_t(w_t-Q_t)}\, e^{-\frac{\cot (\frac{\pi}{4}+t)}{2} (w_t-Q_t)^2} e^{-\frac{\tan (t)}{2}w_t^2}\, F_t  \\ \nonumber
&=& \beta_t \,e^{-i P(x-Q)} 
\,  e^{-\frac{\sin(t)(p-P)^2 + \cos(t)(x-Q)^2 
+2i\sin(t)(x-Q)(p-P)}{2\sqrt{2}\sin(\frac{\pi}{4}+t)}},
\end{eqnarray}
where 
$$
\alpha_t = \left(\mbox{\scriptsize $\sqrt{2}\pi\sin\left(\frac{\pi}{4}+t\right)$}\right)^{-\frac12}\, e^{\frac{\sin{(2t)}}{4}(P^2+Q^2)}\, e^{i\sin^2(t) PQ},
$$
$$
\beta_t = \left(\mbox{\scriptsize $\sqrt{2}\pi\sin\left(\frac{\pi}{4}+t\right)$}\right)^{-\frac12}
$$
and $Y = (P, Q) \in \R^2$, $Q_t = \cos(t)\, Q + i\sin(t)\, P$, $P_t = \sin(t)\, Q + i \cos(t)\, P$.  
\end{definition}

While, here, we have defined the HFrFT on the overcomplete basis of simple Gaussian coherent states, an intrinsic definition of the HFrFT can be  obtained as follows. Namely, let $W\subset L^2(\R,dx)$ be the dense subpace given by the span of the normalized Gaussian coherent states $\{\psi_Y\}_{Y\in \R^2}$. Then, intrinsically, one can define \cite{KMNT}
$$
A_t:W\to {\mathcal H}_t
$$
by
$$
A_t = e^{t\rho(H)} \circ e^{-t\hat H},
$$
where 
$$
\rho(H) = i \left(p\frac{\partial}{\partial x} + x \frac{\partial}{\partial p}\right), \,\,
\hat H = \frac12 \left(-\frac{\partial^2}{\partial x^2}-x^2\right),
$$
are two (first and second order, respectively) differential operators naturally associated to the quantization of the hyperbolic Hamiltonian $H(x,p)= \frac12(p^2-x^2)$, mentioned in the Introduction, realizing a lifting of the canonical transformation in 
imaginary time (\ref{ee-hyp-it}) to the space of quantum states. 

We will also see below in Theorem \ref{equalitytransf} that the relation of the HFrFT with the Segal-Bargmann transform (\ref{slb}), (\ref{benfica}) could also be used to give an intrinsic definition of HFrFT.

\begin{remark}Notice that the level sets of the modulus of the transformed Gaussians on the $(x,p)-$plane are elipses whose major axis, for $t\in (0,\frac{\pi}{4})$, are vertical and decrease in length as $t$ increases, so that, as the elipses become circles, they become equal in length to the horizontal minor axis at $t=\frac{\pi}{4}$. For $t\in (\frac{\pi}{4}, \frac{\pi}{2})$ the major axis becomes horizontal and increases with $t$. As $t\to 0$ or $t\to \frac{\pi}{2}$ the eccentricity converges to $1$. At $t=0$ the level sets are vertical lines while at $t=\frac{\pi}{2}$ they are horizontal lines. 
This is an holomorphic version of and should be compared with the (usual non-holomorphic) FrFT where the 
vertical $p-$axis rotates counter-clockwise to the horizontal $x-$axis through a family of straight lines, as mentioned in the Introduction.
\end{remark}

We then have
\begin{theorem}\cite{KMNT}For each $t\in (0,\frac{\pi}{2})$, $A_t$ is a unitary isomorphism of Hilbert spaces
$$
L^2(\R,dx) \stackrel{A_t}{\cong} {\mathcal H}_t.
$$
\end{theorem}

By letting $t$ approach the right-end of the interval, $t\to \frac{\pi}{2}$, we obtain (up to phase) the usual Fourier transform.

\begin{definition}\cite{KMNT}
\label{fourier}
By allowing $t = \frac{\pi}{2}$ in (\ref{hfrft}) we obtain (and define) for the Gaussian coherent states
\be\label{42.uitreal} 
\left(A_{\frac{\pi}{2}} \, \psi_Y \right) (p, x) :=  \lim_{t\to \frac{\pi}{2}} \left(A_{t} \, \psi_Y \right) (p, x)=
{{\pi}^{-\frac12}}
e^{-i P(x-Q)}\cdot e^{-i(x-Q)(p-P)}\cdot e^{-\frac12 (p-P)^2},
\ee
for $Y = (P, Q)$.
\end{definition}

Consider now the Fourier transform 
\begin{eqnarray}
{\mathcal F}:L^2(\R, dx) &\to& L^2(\R, dp)\\
f&\mapsto& \left({\mathcal F}(f)\right)(p) = \frac{1}{\sqrt{2\pi}}\int_\R e^{ipx}\, f(x) dx.
\end{eqnarray}

By evaluating the Fourier transform of the Gaussian coherente states we establish

\begin{proposition}\cite{KMNT} The transformation $A_{\frac{\pi}{2}}$ coincides, up to a phase, with the Fourier transform, namely, for $f\in L^2(\R,dx)$, we have
$$
\left( A_{\frac{\pi}{2}}(f)\right) (x,p)  = e^{-i px}\left({\mathcal F}(f)\right)(p).
$$
\end{proposition}

Therefore, by taking 
$$
A_0=Id_{L^2(\R,dx)},
$$
we have a one-parameter family of unitary transforms, $\{A_t\}_{t\in [0,\frac{\pi}{2}]}$, such that for the $t\in (0,\frac{\pi}{2})$ the range is, essentially, a Hilbert space of holomorphic functions on the plane and $A_{\frac{\pi}{2}}$ is, up to a phase, the Fourier transform.

\section{Relation to the classical Segal-Bargmann transform}

As we now recall, the original Segal-Bargmann transform \cite{B} is also contained in a one-parameter family of transformations $SB_{s}$, where 
$s\in [0,+\infty)$, with $SB_0=Id_{L^2(\R,dx)}$ \cite{Ha00a,Ha00b,Ha06,Dr,FMMN,KW}.

Consider the complex structure on the plane given by the following holomorphic coordinate
$$
z_{s}= x + is p,  \,\,\, s\in [0,+\infty)
$$
and let $\widetilde{Hol}_{s}$ be the space of $z_{s}$-holomorphic functions on the plane.

Note that, for $\tan (t) = s$, the coordinates $w_t, z_s$ are holomorphic functions of each other, 
$$
w_t = \cos(t) z_{s},
$$
so that indeed, for $t,s$ obeying this relation, $\widetilde{Hol}_{s} = Hol_t$.   
  
Recall that the Segal-Bargmann transforms can be written as \cite{Ha00a,Ha00b, Ha06}
\begin{equation}\label{slb}
SB_{s} = {\mathcal C}_{s} \circ e^{\frac{s}{2} \Delta}, \,\,\ s \in [0,+\infty),
\end{equation}
where $\Delta$ is the Euclidian Laplacian on the plane and where ${\mathcal C}_{s}$ denotes analytic continuation in the $z_{s}$ coordinate.
The transform is therefore obtained by applying the heat kernel operator, which takes functions in $L^2(\R, dx)$ to real-analytic functions, and then by analytically continuing in $z_{s}$.

For $s>0,$ let us consider the Hilbert space 
$$
\tilde {\mathcal H}_{s} = \{f\in \widetilde{Hol}_{s}\,\vert\, \sqrt{s}\int_{\R^2} \overline{f(z_{s})} f(z_{s}) 
e^{-s{p^2}} dxdp<\infty\}. 
$$
 
Then,
\begin{theorem} \cite{B,Ha00a,Ha00b,Ha06,Dr} For each $s >0$, the transform $SB_{s}$ is a unitary isomorphism of Hilbert spaces
$$
L^2(\R,dx) \stackrel{SB_{s}}{\cong} \tilde {\mathcal H}_{s}.
$$
\end{theorem}

 We then have the remarkable identity
 
 \begin{theorem}\label{equalitytransf}
Let $s >0$, $t\in [0,\frac{\pi}{2})$, be such that $\tan{(t)}=s.$ 
For $Y\in \R^2$ recall the Gaussian states $\psi_Y$ in (\ref{42.gau}). Then,
 \begin{equation}\label{benfica}
 A_{t} \left(\psi_Y\right) = (1+s^2)^{\frac14} e^{-\frac{s}{2}p^2}SB_{s} \left(\psi_Y\right).
 \end{equation}
 \end{theorem}
 
 \begin{proof}
 In \cite{KMNT}, we consider a family of transforms $U_\tau$ depending on a parameter $\alpha$. 
 The explicit expression for $SB_{s}(\psi_Y)$ can be obtained from the family of transforms $U_\tau$ by taking a limit 
 $\alpha \to 0$ with $\tau = is$. One then obtains, straightforwardly from \cite{KMNT},
\begin{eqnarray*}
&& SB_{s} \left(\psi_Y \right) (x,p)= 
 \mbox{\scriptsize ${(\pi (1+s))}^{-\frac12}$}\,e^{\frac{s}{2}p^2}\, e^{-iP(x-Q)} e^{-\frac12 \frac{s}{1+s} (p-P)^2 -\frac12 \frac{1}{1+s}(x-Q)^2 -i
 \frac{s}{1+s} (x-Q)(p-P)}.
 \end{eqnarray*}
 The result then follows by direct comparison of the exponents in the Gaussians, using the relation $s = \tan (t)$. 
  \end{proof}

\begin{remark}From the perspective of  \cite{KMNT}, Theorem \ref{equalitytransf} contains a highly non-trivial, and surprising, realization of the Segal-Bargmann transforms. While, as described above, $SB_{s}$ is given by the composition of the (smoothing) heat-kernel operator with the operator of analytic continuation in the variable 
$z_{s}$, the Theorem states that it can also be described by the composition of an unbounded operator (associated to the hyperbolic Hamiltonian $H=\frac12 (p^2-x^2)$)
 and the operator of analytic continuation in the variable $w_t$. 
 \end{remark}

It has been shown \cite{KW} that in the limit $s\to +\infty$, indeed, the Segal-Bargmann transform $SB_{s}$ converges to the Fourier transform (up to a phase).  The family of transforms 
$$
\{SB_{s}\}_{s >0} 
$$
is then also a family of holomorphic fractional\footnote{``Fractional'' in the parameter $t = \arctan (s)$.} Fourier transforms. However, in this case the Fourier transform is reached only in the limit $s \to \infty$.
The geometric context of \cite{KMNT}, however, as we described above, gives a motivation for taking seriously the simple reparametrization
$$
t = \arctan (s),
$$
so that $t\in [0,\frac{\pi}{2}]$, with the Fourier transform now being reached in ``finite time''.

One important point of this work is that while the fractional Fourier transforms of \cite{Co} give a family of transforms mapping $L^2$ signals to $L^2$ signals,
the holomorphic version that we are proposing takes $L^2$ signals to a family of signals which are holomorphic on the plane and, when $t$ reaches the value 
$\frac{\pi}{2}$, one recovers the Fourier transformed signal. Therefore, for $t\in (0,\frac{\pi}{2})$, we get a smoothening of the original signal.

\begin{remark}While in this paper we have, for the sake of simplicity, considered signals in $L^2(\R,dx)$, it is clear that results and formulas generalize straightforwardly to 
$L^2(\R^n,d^nx)$.
\end{remark}

\section{Explicit formulas in the basis of Hermite functions}
\label{shermite}

In this Section we collect useful formulas for the transform $SB_s$ (or, equivalently, for the transform $A_t$ where $s=\tan (t)$) and, in particular, we express it in a basis of Hermite functions.

From \cite{Ha00b}, using the heat-kernel explicitly, one obtains, for $f\in L^2(\R,dx)$,
$$
\left(SB_{s} (f)\right)(z_s)= (2s)^{-\frac12} \pi^{-\frac34}\int_\R e^{-\frac{(z_{s}-x)^2}{2s}} f(x) dx.
$$

For $s>0$, let us consider the orthonormal basis for $L^2(\R,dx)$, $\{h_n^{s}\}_{n\in \N_0}$, given by the Hermite functions
$$
h_n^{s}(x) = a_{s,n}H_n^{s} (x) e^{-\frac{x^2}{4s}}, 
$$
where $a_{s,n}=(2s)^{-\frac14} (\pi n!)^{-\frac12} s^{\frac{n}{2}}$ and the Hermite polynomials are given by (see, for example, \cite{Ha00b,Ab})
$$
H_n^{s} (x) = s^{-n} e^{-\frac{s}{2}\frac{d^2}{dx^2}} \cdot x^n = (-1)^n e^{\frac{x^2}{2s}} \frac{d^n}{dx^n} e^{-\frac{x^2}{2s}} 
= s^{-n} \left(x-s\frac{d}{dx}\right)^n \cdot 1.
$$

Note that since $e^{-\frac{s}{2}\frac{d^2}{dx^2}}$ is the inverse heat operator we have
$$
SB_{s} \left(H_n^{s}\right)  = s^{-n} z_{s} ^n.
$$

\begin{proposition}We have for $s>0$, $n\in \N_0$,
$$
\left(SB_s (h^s_n)\right) (z_s) = d_{s,n} z_s^n e^{-\frac{z_s^2}{6s}},
$$
where $d_{s,n}=(-1)^n a_{s,n} s^{-n} \pi^{-\frac14}  6^{-\frac12}3^{-n}$.
\end{proposition}

\begin{proof}
From the expression for $SB_s$ we obtain,
$$
\left(SB_s(h_n^s)\right)(z_s) = a_{s,n} (2s)^{-\frac12} \pi^{-\frac34} \int_\R e^{-\frac{(z_s-x)^2}{2s}} e^{-\frac{x^2}{4s}} 
\left(\left(x-s\frac{d}{dx}\right)^n\cdot 1\right) dx.
$$
Using 
$$
\left(x+s\frac{d}{dx}\right) e^{-\frac{1}{2s}(\sqrt{\frac32} x -\sqrt{\frac23}z_s)^2} = \left(-\frac13 z_s +\frac{s}{2} \frac{d}{dz_s}\right)e^{-\frac{1}{2s}(\sqrt{\frac32} x -\sqrt{\frac23}z_s)^2}
$$
and integrating by parts one obtains
$$
\left(SB_s(h_n^s)\right)(z_s) = a_{s,n} (2s)^{-\frac12} \pi^{-\frac34} e^{-\frac{z_s^2}{6s}} \left(-\frac13 z_s +\frac{s}{2} \frac{d}{dz_s}\right)^n
\int_\R e^{-\frac{1}{2s}(\sqrt{\frac32} x -\sqrt{\frac23}z_s)^2} dx, 
$$
from which the result follows by evaluating the Gaussian integral. 
\end{proof}

Therefore, given a signal $f\in L^2(\R,dx)$ with expansion
$$
f = \sum_{k=0}^{+\infty} f_n h_n^s, 
$$
we have
$$
\left(SB_s(f)\right)(z_s) = \sum_{n=0}^{+\infty} d_{s,n} f_n z_s^n e^{-\frac{z_s^2}{6s}}.
$$

The inverse formula for $SB_s$ is (see, for example, Theorem 4 in \cite{Ha06})
$$
\left (SB_s^{-1} g\right) (x) =  2^{-\frac12} \pi^{-\frac14} s^{\frac12}\lim_{R\to +\infty} \int_{-R}^R g(z_s) e^{-s\frac{p^2}{2}} dp.
$$

\section*{Acknowledgements}
The authors JM and JPN were partially
by 
FCT/Portugal through the projects UID/MAT/\-04459/2013, PTDC/MAT-GEO/3319/2014, PTDC/MAT-OUT/28784/2017 and by the (European Cooperation in Science and 
Technology) COST Action MP1405 QSPACE. The authors also 
thank the generous support from the Emerging Field Project on Quantum Geometry from Erlangen--N\"urnberg University, 
where this project was initiated.

\providecommand{\bysame}{\leavevmode\hbox to3em{\hrulefill}\thinspace}

\end{document}